\documentclass[numsec, webpdf, modern, medium]{oup-authoring-template}
\onecolumn % for one column layouts

\graphicspath{{Fig/}}

\theoremstyle{thmstyleone}%
\newtheorem{theorem}{Theorem}%  meant for continuous numbers
\newtheorem{proposition}[theorem]{Proposition}%
\theoremstyle{thmstyletwo}%
\theoremstyle{thmstylethree}%

\begin{document}

% \journaltitle{Journal Title Here}
% \DOI{DOI added during production}
% \copyrightyear{YEAR}
% \pubyear{YEAR}
% \vol{XX}
% \issue{x}
% \access{Published: Date added during production}
% \appnotes{Paper}

\journaltitle{ArXiv preprint}
\DOI{https://arxiv.org/abs/2604.21087}
\copyrightyear{2026}
\pubyear{2026}
\vol{XX}
\issue{x}
% \access{Published: Date added during production}
\appnotes{Original paper}

\firstpage{1}

%\subtitle{Subject Section}

\title[Model quality in football]{Model quality in football: Quantifying the quality of an Expected Threat model}

\author[1,$\ast$]{Koen van Arem\ORCID{0009-0004-7658-4580}}
\author[1]{Jakob S\"ohl\ORCID{0000-0002-0831-1714}}
\author[2]{Mirjam Bruinsma}
\author[1]{Geurt Jongbloed\ORCID{0000-0003-4708-5868}}
% \author[4]{Fifth Author\ORCID{0000-0000-0000-0000}}

\address[1]{\orgdiv{Delft Institute of Applied Mathematics}, \orgname{Delft University of Technology}, \orgaddress{\country{The Netherlands}}}
\address[2]{\orgdiv{Football Analytics}, \orgname{AFC Ajax}, \orgaddress{\country{The Netherlands}}}
% \address[3]{\orgdiv{Department}, \orgname{Organisation}, \orgaddress{\street{Street}, \postcode{Postcode}, \state{State}, \country{Country}}}
% \address[4]{\orgdiv{Department}, \orgname{Organisation}, \orgaddress{\street{Street}, \postcode{Postcode}, \state{State}, \country{Country}}}

\corresp[$\ast$]{Corresponding author. \href{email:k.w.vanarem@tudelft.nl}{k.w.vanarem@tudelft.nl}}

% \received{Date}{0}{Year}
% \revised{Date}{0}{Year}
% \accepted{Date}{0}{Year}

% \editor{Associate Editor: Name}

%\abstract{
%\textbf{Motivation:} .\\
%\textbf{Results:} .\\
%\textbf{Availability:} .\\
%\textbf{Contact:} \href{name@email.com}{name@email.com}\\
%\textbf{Supplementary information:} Supplementary data are available at \textit{Journal Name}
%online.}

\abstract{The recent growth in data availability in football has increased the risk of incorrect use of data-driven models, making guidelines on their validation and application necessary. The Expected Threat (xT) model is an accessible option for football organisations that start building in-house methods, yet little is known about how to assess its quality. The aim of this study is twofold: to examine how the model quality depends on the number of game states and the number of training points, and to translate these results into guidelines for constructing and applying the model. Using the Markov chain underlying the model, we perform theoretical analyses and simulations to study the estimation error. These show that the estimation error is approximately lognormal for a given number of training points and game states. Additionally, we combine the simulations with expert consultation to establish the estimation error beyond which player evaluations based on the Expected Threat model become unreliable for scouting applications. From this, we derive rules of thumb to ensure the quality of an Expected Threat model before application, and we illustrate through an example how a validated model can be applied in practice. Because the approach generalises to Expected Possession Value models, this paper illustrates a framework to systematically quantify model quality, despite the ground truth being unobservable in football analytics.} 

\keywords{Sports analytics, Association football/Soccer, Application guidelines, Error quantification, Possession Value Model, Markov chains, Data-driven decision-making, Football scouting}

% \keywords[Abbreviations]{abbreviation1, abbreviation2, abbreviation3, abbreviation4}
\keywords[Abbreviations]{xT: Expected Threat, xG: Expected Goal, EPV: Expected Possession Value}

% \otherabstract[Additional Abstract]{Use this element for elements such as Graphical abstract, Lay summary, Translated abstract etc. Que cum aut etum qui ium dolupta ssequia autati odis demporepe ad et es alit rem repudaerae min et volorum re volupta nobit volectur aut fuga.}

% \otherabstract[Graphical Abstract]{\colorbox{black!20}{\hbox to 0.97\textwidth{\vbox to 50pt{}}}}

% \boxedtext{Key Messages}{
% \begin{itemize}
% \item Key boxed text here.
% \item Key boxed text here.
% \item Key boxed text here.
% \end{itemize}}

\maketitle

%\begin{epigraph}
%Epigraph text. Ximporem qui reperov idempedit modio. Bisto imagnatem quae aceptis
%nobitae quid eum rae adignis quias-sit vellacc uptatur sunt quis rentis eaquasit alia deliquam
%rec-to consed unt. Empor sum ratur ressimusdae. Nam fugiae.
%\source{Epigraph source}
%\end{epigraph}

\section{Introduction}
Although the availability of data in (association) football has grown substantially over the last decade, there are few clubs that have in-house models to improve decision-making in their scouting or match analysis processes \citep{Iacono2025}. This, despite the potential impact ranging from improved talent identification to supporting tactical analyses \citep{Kholkine2026, Teixeira2025, RicoGonzlez2023}. For football clubs starting to build in-house models, accessible models can be more easily integrated into their processes. However, the use of easily accessible methods poses the risk of incorrect use of data-driven techniques \citep{OlthofDavis2025}. Post hoc argumentation and drawing conclusions that might be based on insignificant differences can be possible pitfalls. Consequently, it is important that these models are accompanied by guidelines on their validation and how to use them \citep{Iacono2025, OlthofDavis2025}.

A data-driven model that has generally been embraced by practitioners at football organisations is the Expected Goal model, often shortened as xG \citep{Robberechts2020, Anzer2021, Mead2023}. It models the probability of a shot being scored given a specific in-game situation. This can be used to evaluate shooting performance, goalkeeper performance, or scoreline analysis. The Expected Threat model \citep{Rudd2011} generalises the xG model. The Expected Threat (xT) model describes the probability of scoring a goal in the same possession chain given the current game state. In the traditional formulation, it models the progression of the ball during a possession chain because it defines the game state by the position of the ball during a possession. The Expected Threat model makes it possible to evaluate the quality of general open-play actions, whereas xG only considers actions related to shots. These action evaluations provide information for opponent analysis or player scouting. The framework uses procedures that are intuitive for practitioners \citep{VanRoy2020} and can be visualised in its traditional low-dimensional setting as illustrated by \citet{Singh2018}. 
This makes the Expected Threat model very accessible for football clubs starting to incorporate data into their decision processes.

Despite its accessibility, little is known about the quality of the model's estimates. Differences between model and the game of football in reality arise both from modelling assumptions and from the estimation process. In this paper, we focus on studying issues related to the latter. More specifically, it is unclear how the number of different game states, indicative of the model complexity, and the size of the training set influence the errors of the estimations by the model. Consequently, practitioners face challenges during model construction, application, and interpretation of the results. 

The aim of this paper is twofold: (1) to obtain insights into the dependence of the error of the Expected Threat model on the number of game states and the number of training points, and (2) to provide guidelines for practitioners to build and apply an Expected Threat model. 

To this end, we first give a mathematical description of the model. We then use this to find theoretical results on the error of the model, which provide a theoretical guaranteed error bound and understanding of the model behaviour. Simulations provide additional insights into the error of the model for a specified training sample size and number of game states. We find an approximate distribution of the estimation error, which suggests tighter, hypothesised bounds. Additionally, we perform simulations to find what size of the estimation error is acceptable for applications in player scouting. Finally, we use these results to provide guidelines to obtain a validated model and we illustrate how such a model can be applied to compare player performances at the 2020 European Championship.

\section{Expected Threat model}
\subsection{Theoretical framework}
We first discuss the traditional Expected Threat model to establish a foundation for the work in this paper. The Expected Threat, $xT(s)$, is the probability that a goal is scored in the current possession chain in football, given the current game state $s$. The concept was first introduced by \citet{Rudd2011} and later popularised by \citet{Singh2018}. In the traditional formulation, the game state is defined by the location of the ball during a possession. The possession chains in the game are described by a Markov chain. Within this framework, the probability of scoring given the current state is well-defined and can be estimated.

Because the aim in football is to score goals, a good pass or dribble increases the likelihood of scoring. Such an action $a_i=(s_{i-1}, s_i)$ moves the game from state $s_{i-1}$ to $s_i$. The quality of an on-the-ball action is then described in \eqref{eq: diff in xT} as the difference in Expected Threat before and after an action.
\begin{equation}\label{eq: diff in xT}
    \Delta xT(a_i) = xT(s_i) - xT(s_{i-1}), \qquad \text{where } a_i=(s_{i-1}, s_i).
\end{equation}

The Expected Threat model assumes that football adheres to the Markov assumption: distribution of the future states in a football game only depends on the current state and not on the past. When the game is in state $s$, a goal could be scored if the player with the ball takes a shot, and the shot is successfully scored. A goal can also be scored by transitioning to another game state $s'$ and scoring from that position. The probability of scoring from that position is again the Expected Threat at state $s'$: $xT(s')$. Consequently, the Expected Threat model should satisfy the equation as given in \eqref{eq: characteristic equation detailled}.
\begin{equation}\label{eq: characteristic equation detailled}
    xT(s) = P(shot|s)\cdot xG(s) + \sum_{s'\in S} T(s, s') \cdot xT(s'),
\end{equation}
where $xG(s)$ is the probability of scoring a goal given a shot from state $s$, also known as the extensively-studied Expected Goal (xG) value, $S$ is the set of all game states and $T(s, s')$ is the probability of the game transitioning from state $s$ to state $s'$. Alternatively, the $xT$ can be seen as the probability of the Markov chain being absorbed into the implicit game state of a goal being scored.

When the set of game states $S$ is of finite size $M$, the Expected Threat model allows for an alternative notation with vectors and matrices. To do this, the Expected Threat is denoted as $xT\in [0,1]^M$. The probability of scoring a goal can be contained in one vector $g\in [0,1]^M$, which is element-wise defined as $g(s) = P(shot|s)\cdot xG(s)$. The transition matrix $T\in[0,1]^{M\times M}$ gives the probabilities from each in-game state to each other in-game state. The defining equation can then be reformulated as 
\begin{equation}\label{eq: characteristic equation vectorized}
    xT = g + T\cdot xT.
\end{equation}

\subsection{Estimation}\label{subsec: xT model: estimation}
In practice, the quantities in \eqref{eq: characteristic equation vectorized} are unknown and have to be estimated. In recent years, the quantity of collected data in football has increased, resulting in often proprietary data sets \citep{Rein2016, Herold2019, OlthofDavis2025}. With historical data of in-game actions, the values of $P(shot|s)$, $xG(s)$, $T_{s,s'}$,  or the vector $g$ and matrix $T$ can be estimated. This is generally done by counting occurrences and taking the empirical mean.

The only remaining unknowns are the $xT$-values, which are solutions to a linear system of equations described by \eqref{eq: characteristic equation vectorized}. This can be solved with various methods, but a specific version of the Value Iteration algorithm \citep{Bransen2025} is used that is intuitive for coaches and players. 

This Value Iteration algorithm calculates the probability of scoring within the next $k$ actions. Scoring a goal within one action can only be done by directly shooting and scoring successfully. The initial value is therefore set as $P(shot|s)\cdot xG(s)$ or $xT^{(1)} = g$. Scoring within $k$ actions is possible by scoring directly or moving the ball and then scoring a goal within $k-1$ actions. This allows the recursive calculation of the probability of scoring within $k$ actions, and this algorithm is summarised in Algorithm~\ref{alg: iterative algorithm vectorized}. The values of this algorithm converge to the correct values when $k\to\infty$, as described in Proposition~\ref{prop: truncation error vectorized}.
\begin{algorithm}[hbt!]
    \caption{Vectorised Value Iteration algorithm for calculating $xT$ values}
    \label{alg: iterative algorithm vectorized}
    
    \begin{algorithmic}[1]
        \State \textbf{Input:} Convergence threshold $\varepsilon > 0$
        \State \textbf{Known values:} $g \in [0,1]^M$, $T \in [0,1]^{M \times M}$
        \State \textbf{Unknown:} $xT \in [0,1]^M$
        \State $xT^{(1)} \gets g$, $k \gets 1$
        \While{$\| xT^{(k)} - xT^{(k-1)} \| > \varepsilon$}
            \State $xT^{(k+1)} \gets g + T\cdot xT^{(k)}$
            \State $k \gets k + 1$
        \EndWhile
        \State \Return $xT^{(k)}$
    \end{algorithmic}
\end{algorithm}

\subsection{Definition of game state}\label{subsec: definition of game state}
An important aspect of the Expected Threat model is how the game states are defined. Traditionally, the field is discretised and divided into different rectangles. The rectangle in which the ball-possessing player is located during a possession chain defines the game state \citep{Singh2018}. A common discretisation of the pitch is a $16\times12$ grid. This grid roughly follows the ratio between the dimensions of a football pitch, and it results in a model with $M=16\cdot 12=192$ game states. In this way, the discretisation of the pitch in two dimensions determines the game states of the model.

The discretisation grid is a tuning parameter of the Expected Threat model. A model with a finer grid is better able to distinguish between different in-game situations and could be considered favourable. On the other hand, a finer grid results in more probabilities to estimate as the sizes of $g$ and $T$ grow. With a fixed size of the training data set, this decreases the number of data points that can be used to estimate a single probability and reduces the quality of the estimation. As a result, the choice of the grid is a model tuning parameter that strongly influences the model precision and quality of the estimates.

The traditional approach of defining game states solely by the ball-possessing player's positions does not capture the full game dynamics that influence scoring probability \citep{VanRoy2020}. \citet{vanArem2024} extended the traditional model by adding two extra dimensions: an indicator for a high or low ball and the number of defenders between the player and the goal. Because this rapidly increases the number of game states, they estimated the xG-values via an existing xG-model and used kernel density estimation to prevent overfitting of the transition matrix. \citet{HassaniRamdaniLofti2025} added the xG-value as an extra dimension to the Expected Threat model. Based on the expected behaviour of football players, they assumed extra structure on the transition probabilities. In this way, the Expected Threat model can be adjusted and extended to better capture game dynamics.

In a broader sense, the Expected Threat model is an Expected Possession Value (EPV) model. EPV models are often studied with continuous state spaces, and scoring probabilities are generally estimated through neural networks such as (convolutional) neural networks \citep{Fernandez2019, Fernandez2021}, graph neural networks \citep{Stockl2021}, or autoencoder-decoders \citep{Overmeer2025}. These models can be used to assess decision-making \citep{VanRoy2023, Pulis2022, Rahimian2024} and assess defensive quality \citep{Kim2026}. Hereby, they can provide extra insights for clubs with a dedicated data department.

However, these more complex approaches require tracking data, which is less widely available. This makes it harder to apply those models in scouting processes, where tracking data is less accessible. Additionally, the described methods perform smoothing or regularisation to deal with the increased number of the game states, which complicates the mathematical analysis of the estimation error. Therefore, this paper focuses on the traditional formulation of the Expected Threat model, although most of it is applicable to the more general class of EPV models.

\subsection{Quantifying the model quality}\label{subsec: quantification of estimation error}
The Expected Threat model aims to approximate the game of football as closely as possible. There are two factors that may cause a difference between model and reality.
First, the model makes assumptions such as the Markov assumption and the definition of the game states. The game state may ignore relevant information such as player locations, the threat on the field changes gradually instead of with few discrete steps, players are influenced by the past, albeit slightly,. This causes a discrepancy due to the modelling assumptions.

Second, an error is made within the context of the assumed model. Estimating $\hat{g}$ and $\hat{T}$ and applying Algorithm~\ref{alg: iterative algorithm vectorized} can introduce deviations from the true values. Because this error arises from the estimation process within the context of the assumed model, we refer to it as the estimation error.
Both influence how closely the Expected Threat model describes reality. Although we acknowledge that a finer grid decreases the influence of the model assumptions, the focus of this paper is to study the estimation error, since little is known about this aspect.

To study the size of the estimation error, we assume that there exist true $xT$-values for each set of game states $S$ of size $M$. These values are the actual probabilities of scoring a goal in each defined game state. When training an Expected Threat model, these quantities are estimated by taking empirical means and applying Algorithm~\ref{alg: iterative algorithm vectorized}. The resulting estimation error is $xT-\widehat{xT}^{(k)}\in[0,1]^M$, where $xT$ represents the true value and $\widehat{xT}^{(k)}$ the estimated value after $k$ iterations. The size of this error can then be measured with any vector norm: $\|xT-\widehat{xT}^{(k)}\|$.

In the context of this paper, the $\ell_\infty$-norm, defined as $\|xT-\widehat{xT}^{(k)}\|_\infty = \max_{s\in S}|xT(s)-\widehat{xT}^{(k)}(s)|$, is a suitable candidate for three reasons. First, this norm does not require rescaling to compare models among different number of game states $M$. Second, the probability of scoring a goal is close to 0 far away from the goal, and the error in estimating these $xT$-values is low in practice. By taking the maximal values, many small errors do not deflate the measured error when significant errors are made at critical states close to the goal. Finally, the $\ell_\infty$-norm provides an intuitive way of measuring the error of the estimated values as it describes the maximal error for each single state. This gives practitioners a clear handle to interpret the margin of error when applying the model at football clubs. The size of the estimation error is therefore measured as $\|xT-\widehat{xT}^{(k)}\|_\infty$.

\section{Theoretical results}\label{sec: theoretical results}
The goal of training an Expected Threat model is to estimate the true underlying $xT$-values, but errors can arise in the estimation process. There are two parts of the process that can cause errors. First, the statistical estimation of the vector $g$ and in-game transition matrix $T$ causes an error. A second source is the numerical error created by Algorithm~\ref{alg: iterative algorithm vectorized}, which calculates the estimated $xT$-values from the estimated $g$ and $T$. This can be summarised as 
\begin{equation}\label{eq: split of the error}
    xT-\widehat{xT}^{(k)} = \left(xT - \widehat{xT}\right) + \left(\widehat{xT} - \widehat{xT}^{(k)}\right),
\end{equation}
where $xT$ denotes the vector with the true underlying values, $\widehat{xT}^{(k)}$ the estimated values after $k$ iterations through Algorithm~\ref{alg: iterative algorithm vectorized} based on a sample of size $N$, and $\widehat{xT}$ the statistical estimate without any error caused by Algorithm~\ref{alg: iterative algorithm vectorized}. 

To analyse the model, we introduce certain assumptions about the transition matrix and the estimation process. For transition matrices in Markov chains, it holds that $0\leq\|T\|_\infty\leq1$. Since the Expected Threat model contains two implicit non-game states, scoring a goal and losing possession, $T$ is a submatrix of the full transition matrix. Consequently, $\|T\|_\infty<1$ holds whenever every game state has nonzero probability of either outcome. As any action could have a bad execution, the assumption $\|T\|_\infty<1$ is reasonable. The same assumption on $\|\hat{T}\|_\infty$ is justified for large enough sample sizes since it estimates $T$. If it fails, the corresponding game state can be removed. By concatenating incoming and outgoing actions, an equivalent model satisfying the assumption is obtained.

Lastly, we consider a dataset of $N$ in-game actions, for instance, all actions of a season in a league. Let $p_g\approx0.02$ denote the proportion of shots, and let $\bar{N}_g=p_gN/M$ and $\bar{N}_T=(1-p_g)N/M$ denote the mean number of shots and ball-moving actions per game state, respectively. We assume both are uniformly divided over all game states, as this yields interpretable results while still capturing the property that the $\bar{N}_g$ and $\bar{N}_T$ grow proportionally with $N$. 

\subsection{Numerical error}\label{subsec: value iteration algorithm}
The term $\widehat{xT} - \widehat{xT}^{(k)}$ in \eqref{eq: split of the error} can be considered the numerical error caused by Algorithm~\ref{alg: iterative algorithm vectorized} after $k$ iterations. The following result gives a bound on this numerical error. Although it is formulated with the estimated quantities concerning $\widehat{xT}$, it also holds when replacing them by the true underlying values like $xT$ and $xT^{(k)}$. 

\begin{proposition}\label{prop: truncation error vectorized}
    Consider an estimated Expected Threat model as a Markov chain. Assume that $\|\hat{T}\|_\infty<1$. Then, the error of the Value Iteration algorithm of the Expected Threat model (Algorithm~\autoref{alg: iterative algorithm vectorized}) after $k$ iterations satisfies
    \begin{equation}\label{eq: bound truncation error iterative algorithm}
        \left\| \widehat{xT}-\widehat{xT}^{(k)}\right\|_\infty\leq \frac{\|\hat{g}\|_\infty\ \|\hat{T}\|_\infty^k }{1-\|\hat{T}\|_\infty}.
    \end{equation}
\end{proposition}

As a result of Proposition~\ref{prop: truncation error vectorized}, the error decreases geometrically at each iteration of Algorithm~\autoref{alg: iterative algorithm vectorized} with a factor $||\hat{T}||_\infty$. In practice, the transition matrix is sparse, and the algorithm can be implemented such that iterations are relatively cheap. Consequently, the number of iterations can generally be chosen such that the numerical error is negligible.

\subsection{Statistical Error}
Since the numerical error can be made arbitrarily small by performing sufficiently many iterations, the estimation error mainly arises from the statistical estimation of the quantities. The statistical error splits into the estimation error of the vector $g$ and the matrix $T$. A probabilistic bound on these statistical errors can be obtained via concentration inequalities. These insights can then be combined to find one probabilistic bound on the size of the statistical error.

\subsubsection{Split of statistical error}
The size of the statistical error can be split into a term representing the error in the goal probability and one representing the error in the transition probabilities. This is described in the following proposition.
\begin{proposition}\label{prop: split of statistical error}
    Let $xT$, $g$, and $T$ denote the quantities of the true Expected Threat model, and let $\widehat{xT}$, $\hat{g}$, and $\hat{T}$ denote the quantities of the estimated Expected Threat model. If $\|T\|_\infty<1$, then 
    \begin{equation}\label{eq: split of statistical error}
        \|xT - \widehat{xT}\|_\infty \leq \frac{\|g - \hat{g}\| _\infty+ \|(T - \hat{T})\widehat{xT}\|_\infty}{1 - \|T\|_\infty}.
    \end{equation}
\end{proposition}

The proposition shows that the statistical error consists of the error in the estimation of model parameters $g$ and $T$ scaled by the factor $\frac{1}{1-\|T\|_\infty}$. The error of the transition probabilities is propagated to the statistical error, but only after being weighted by the estimated $xT$-values of the corresponding states. Consequently, errors in transition probabilities are most problematic if they occur in states with a high estimated $xT$-value. 

The rescaling factor $\frac{1}{1-\|T\|_\infty}$ is large if $\|T\|_\infty\approx1$. This happens if a game state has a very small probability of both losing ball possession and scoring a goal. In such a situation, the bound will increase because the $xT$-values are difficult to estimate for this Markov chain. 

\subsubsection{Concentration inequalities}
The model parameters $g$ and $T$ are estimated via the empirical mean. Concentration inequalities can be used to derive bounds on the estimation error of both. This results in bounds in terms of the number of game states $M$ and the size of the training data set $N$. Combined, they provide a bound on the statistical error and can give insights into the impact of errors during the estimation of the model parameters $g$ and $T$.

The goal-scoring vector, $g$, gives the probability of scoring given ball possession in game state $s$. For each state, it holds that $P(goal|s)=P(goal|shot,s)\cdot P(shot|s)$. The estimate $\hat{g}$ is often obtained by estimating the probabilities of a shot being taken, and the shot being successful with the empirical mean. These are then multiplied to find the value $\hat{g}(s)$. Equivalently, one could estimate $g$ by taking the empirical mean of the goals being scored from all visits to a state. The following proposition gives a bound on the error of such estimation.

\begin{proposition}\label{prop: statistical error goal vector}
    Suppose that each entry of $g$ is estimated by the empirical mean of $\bar{N}_g$ independent Bernoulli random variables with expectation $g(s)$. Then, the bound 
    \begin{equation}
        \|g-\hat{g}\|_\infty < \sqrt{\frac{\log(2M/\alpha)}{2\bar{N}_g}} = M^{1/2}\sqrt{\frac{\log(2M/\alpha)}{2N}}
    \end{equation}
    holds at least with probability $1-\alpha$.
\end{proposition}

Proposition~\ref{prop: statistical error goal vector} assumes that the observed transitions leaving each state $s$ are independent random variables. Note that this is required within each state only and independence between samples from different states is not assumed. Because the Markov assumption gives that different visits to a specific state are independent, this assumption holds for the Expected Threat model.

The transition matrix gives the probabilities of moving from one in-game state to another in-game state. These are estimated by taking the empirical mean for each individual transition probability. The following proposition gives a bound on the error of the estimated transition matrix. 

\begin{proposition}\label{prop: statistical error transition matrix}
    Suppose that each entry of $T$ is estimated by taking the empirical mean of $\bar{N}_T$ independent Bernoulli random variables with expectation $T({i,j})$. Then, the bound
    \begin{equation}
        \|T-\hat{T}\|_\infty < M\sqrt{\frac{\log(2M^2/\alpha)}{2\bar{N}_m}} = M^{3/2}\sqrt{\frac{\log(2M^2/\alpha)}{2N}}
    \end{equation}
    holds with probability at least $1-\alpha$.
\end{proposition}

Proposition~\ref{prop: statistical error goal vector} shows that the bound on the error in estimating $g$ grows with order $M^{1/2}\sqrt{\log(M)}$. According to Proposition~\ref{prop: statistical error transition matrix}, the error in estimating the transition matrix is bounded by a quantity growing with order $M^{3/2}\sqrt{\log(M)}$. The error in estimating the transition matrix will grow faster than the error for estimation of the goal-scoring probabilities in terms of $M$. 

The influence of $N$ shows the same behaviour for both model parameters. Both bounds decrease with an order of $1/\sqrt{N}$. As expected, the estimation of the model parameters is better for large values of $N$.

\subsubsection{Estimation error}\label{subsubsec: theoretical part estimation error}
The derived results can be combined to obtain one bound on the error of the estimated Expected Threat values. The following theorem is a direct consequence of the derived propositions. As discussed in the beginning of this section, most assumptions for this theorem are naturally satisfied during application. 

\begin{theorem}\label{thm: estimation error}
    Let $xT$, $g$, and $T$ denote the quantities of the true Expected Threat model with $M$ game states, and let $\widehat{xT}^{(k)}$, $\hat{g}$, and $\hat{T}$ denote the quantities of the estimated Expected Threat model with $k$ iterations of Algorithm~\autoref{alg: iterative algorithm vectorized}. Assume that $\|T\|_\infty<1$ and $\|\hat{T}\|_\infty<1$. Then, with probability $1-\alpha$, the error of the estimated model is bounded by
    \begin{equation}
    \begin{split}
        \left\|xT-\widehat{xT}^{(k)}\right\|_\infty &< \frac{1}{1-\|T\|_\infty}\color[HTML]{848482}
        \underbrace{\color{black}
            \left(\color[HTML]{848482}
            \overbrace{\color{black}\frac{1}{\sqrt{p_g}}M^{1/2}\sqrt{\frac{\log(2M/\alpha)}{2N}}}^{\text{Error } \hat{g}}
            \color{black}+
            \color[HTML]{848482}\overbrace{\color{black}\frac{1}{\sqrt{1-p_g}}M^{3/2}\sqrt{\frac{\log(2M^2/\alpha)}{2N}}}^{\text{Error } \hat{T}}
            \color{black}\right)
        }_{\text{Statistical error}}\\
        &+ 
        \color[HTML]{848482}\underbrace{\color{black}\frac{\|\hat{g}\|_\infty\ \|\hat{T}\|_\infty^k}{1-\|\hat{T}\|_\infty}.}_{\text{Numerical error}}
    \end{split}
    \end{equation}
\end{theorem}

The two terms corresponding to the errors of $\hat{g}$ and $\hat{T}$ can be combined into one by taking the highest order terms. As discussed in Section~\ref{subsec: value iteration algorithm}, the numerical error is negligible in typical situations and the last term of the bound can be ignored. The fact that $p_g \ll 1$ gives that $\frac{1}{\sqrt{1-p_g}}\approx 1$. These insights can be combined to find the following approximate bound with probability $1-\alpha$.
\begin{equation}
    \|xT-\widehat{xT}^{(k)}\|_\infty\ \lesssim\ \frac{p_g^{-1/2}+1}{1-\|T\|_\infty} M^{3/2}\sqrt{\frac{\log(2M^2/\alpha)}{2N}}.
\end{equation}
This bound shows that the error in terms of $N$ and $M$ is at most of order $O(M^{3/2}\sqrt{\log(M)}/\sqrt{N})$ with high probability. It shows that the estimation error is expected to grow when the number of game states $M$ is increased or the number of training data points $N$ is decreased.

However, the errors of transition probabilities and the goal probabilities have different multiplication factors. Because $p_g$ is small, the term with a factor $M^{1/2}\sqrt{\log(M)}$ term gets a higher multiplicative constant. For finite values of $M$, the error in goal-scoring probabilities $\|g-\hat{g}\|_\infty$ can be dominant instead of the error in the transition probabilities $\|T-\hat{T}\|_\infty$, which has a higher order. 

The interpretation of the game of football also provides extra nuance. States with a low threat tend to be visited more often. These low-threat states often have a low probability of scoring directly and many data points, leading to good estimation of the values $g$. Similarly, the error in estimating $T$ is weighted by $\widehat{xT}$. This means that low-threat game states contribute little to the estimation error. Because only a few game states are of high threat, most game states typically have small estimation errors.

Because many game states have small errors and the error in estimating $\hat{g}$ might be dominant for finite values of $M$ and $N$, we hypothesise that the error of the Expected Threat model is lower in typical applications. The orders might grow less fast than the worst-case $O(M^{3/2}\sqrt{\log(M)}/\sqrt{N})$. Although the theoretical results provide insights into the behaviour of the model, the typical behaviour of the estimation error during application is unknown.

\section{Simulations}\label{sec: Simulations}
Fortunately, the Markov chain underlying the Expected Threat model makes it possible to study the estimation error via simulation. Therefore, we use simulations to describe the estimation error realistically and complement the theoretical results.

Two types of simulations are performed. First, we perform a simulation study to obtain an approximate distribution of the estimation error given realistic values of both the training sample size $N$ and the number of game states $M$. Second, we combine expert consultation with simulations to find which estimation errors result in reliable player ratings for scouting purposes. These simulations provide information on the size of the estimation error and its influence on the quality of player ratings. With this, we will formulate rules of thumb to constrain the error of the Expected Threat model in the context of application to player scouting in Section~\ref{sec: Application}.

\subsection{Methods}\label{subsec: simulation distribution estimation error methods}
For a specific grid of size $M$, we train an Expected Threat model on a reasonably large set of event data. For the purpose of the analysis, these models are assumed to be the ground truth. 

\begin{figure}[h]
    \centering
    \includegraphics[width=0.95\linewidth]{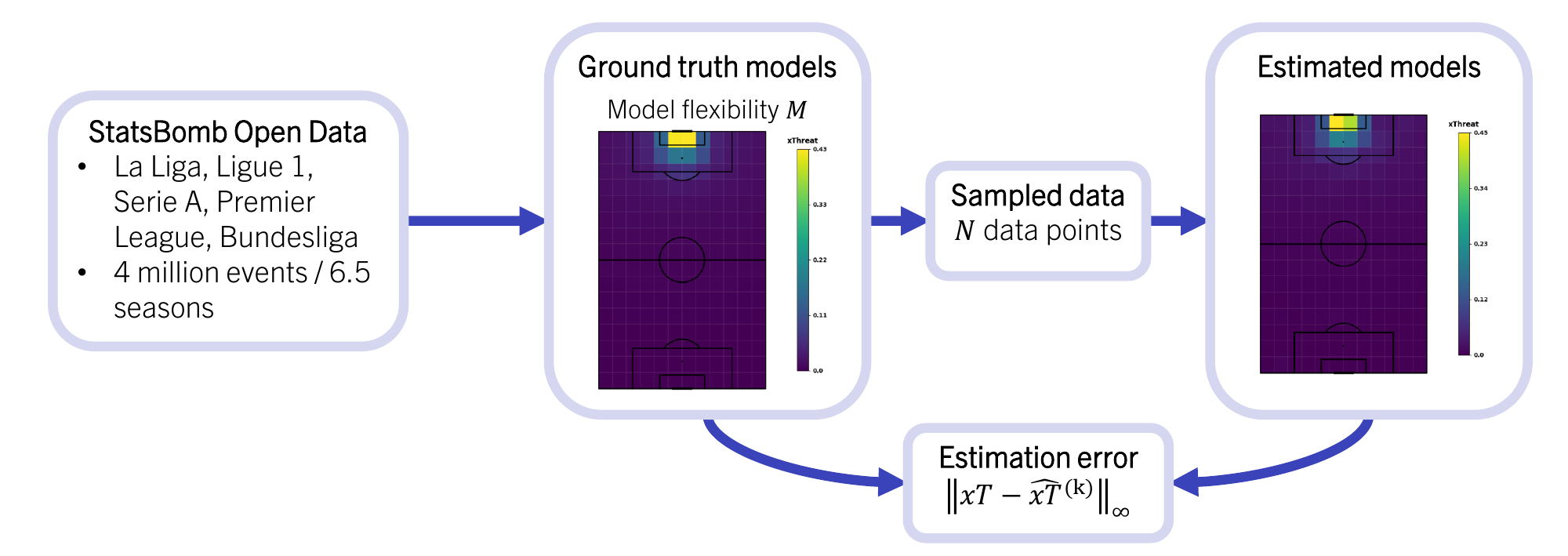}
    \caption{A visualisation of the simulation scheme used in the study.}
    \label{fig: simulation scheme}
\end{figure}

Based on the Markov chain behind the ground truth model, a new data set of size $N$ can be simulated by sampling different possession chains from the ground truth model. A new model with the same grid can be trained, which estimates the $xT$-values of the original model. As both the ground truth $xT$ and the estimated $\widehat{xT}$ are known, the size of the estimation error can be calculated for the estimated model. The simulated estimation error is a bootstrap approximation of the estimation error. This process is visualised in Figure~\ref{fig: simulation scheme} and yields a situation where the error of the model can be assessed despite the general lack of ground truth in sports modelling \citep{DavisBransen2024}.

To construct the ground truth models, we use the openly available StatsBomb open data \citep{StatsBombOpenData}. All available games from the Premier League, Ligue 1, Serie A, La Liga, and the Bundesliga are used. The events that do not describe passes, dribbles, errors, clearances, or shots are filtered out. This results in a dataset of approximately 4 million in-game events, or roughly 6.5 full seasons of a single league.

The code for the simulations is available on \href{https://github.com/kvanarem42/xthreat-error-quantification}{GitHub}. The computations are performed on the DelftBlue Supercomputer \citep{DHPC2024}.

\subsubsection{Distribution of the estimation error}
We can use the simulation framework to investigate the behaviour of the estimation error in practice by creating a data set of estimation errors with corresponding values of $N$ and $M$.
We fit a statistical model to the simulated estimation error to better describe its behaviour, improving on the theoretical bounds. 

Because the size of the estimation error, $\|xT-\widehat{xT}^{(k)}\|_\infty$, is a maximum of absolute values, the statistical model must have positive support. The theoretical analysis indicates that the error becomes larger for increasing values of $M$ and decreasing values of $N$, so both the variance and mean can be expected to depend on these variables. However, the theoretical bound is a worst-case result, and we hypothesise that the true dependence on $M$ and $N$ is in reality of lower order than $O(M^{3/2}\sqrt{\log(M)}/\sqrt{N})$. This motivates a statistical model with a multiplicative structure and flexible powers of $M$ and $N$.

Therefore, a parametric model with a lognormal distribution is assumed as as follows. \begin{equation}\label{eq: simulations lognormal distribution}
    \left\|xT-\widehat{xT}^{(k)}\right\|_\infty = C M^\alpha/(\sqrt{N})^\beta e^\varepsilon, \qquad \text{where }\varepsilon\sim N(0, \sigma^2).
\end{equation} 
This random variable only attains positive values. The parameters describing the powers of $N$ and $M$ make it possible to scale both the mean and expected value of the error. Additionally, they allow for inference on the powers for $N$ and $M$. 

When $c=\log(C)$, this is equivalent to
\begin{equation}\label{eq simulations lognormal distribution logarithm}
    \log\left(\left\|xT-\widehat{xT}^{(k)}\right\|_\infty\right) = c + \alpha \log(M) - \beta\log(\sqrt{N}) + \varepsilon, \qquad \text{where, } \varepsilon\sim N(0, \sigma^2).
\end{equation}
This formulation gives a linear model in the transformed features with homoscedastic normal residuals. We can use ordinary least squares (OLS) to estimate the unknown parameters $c, \alpha, \beta$, and $\sigma^2$. The fit of the estimated distribution is then assessed based on the $R^2$, the normality and homoscedasticity of the residuals. 

We can use the simulation framework to investigate the behaviour of the estimation error in practice. The sampling process is repeated 1,000 times for different values of $N$ and $M$. The combinations of the values of $N$ are given in Table~\ref{tab: values N simulation distribution} and those of $M$ with corresponding grid dimensions in Table~\ref{tab: values M simulation distribution}. The grids roughly follow the same proportions as the traditional $16\times12$ grid. These simulations result in 104,000 estimated models with different values of $M$ and $N$, with their corresponding error quantities. 

\begin{figure}[htb]
\begin{minipage}[t]{0.38\textwidth}
    \centering
    \captionsetup{width=\textwidth}
    \captionof{table}{Sample sizes $N$ used in the simulations to study the distribution of the estimation error.}
    \label{tab: values N simulation distribution}
    \begin{tabular}{ccc}
        \hline
        \multicolumn{3}{c}{\textbf{Sample sizes \(N\)}} \\
        \hline
        100,000   & 240,000 & 1,300,000 \\
        130,000   & 370,000 & 4,000,000 \\
        170,000   & 630,000 &  \\
        \hline
    \end{tabular}
\end{minipage}
\hfill
\begin{minipage}[t]{0.58\textwidth}
    \centering
    \captionsetup{width=\textwidth}
    \captionof{table}{Grid sizes and number of game states $M$ used in the simulations to study the distribution of the estimation error.}
    \label{tab: values M simulation distribution}
    \begin{tabular}{ccc}
        \hline
        \multicolumn{3}{c}{\textbf{Grid dimensions: $m_x\times m_y$ ($M$)}} \\
        \hline
        $8\times6\phantom{00}\ (48)\phantom{0}$   & $20\times15\ (300)\phantom{0}$ & $48\times36\ (1728)$ \\
        $10\times8\phantom{0}\ (80)\phantom{0}$   & $24\times18\ (432)\phantom{0}$ & $56\times42\ (2352)$ \\
        $12\times9\phantom{0}\ (108)$   & $28\times21\ (588)\phantom{0}$ & $64\times48\ (3072)$ \\
        $14\times11\ (154)$ & $32\times24\ (768)\phantom{0}$& \\
        $16\times12\ (192)$ & $40\times30\ (1200)$ & \\
        \hline
    \end{tabular}
\end{minipage}
\end{figure}

Some combinations of $M$ and $N$ result in errors too large for practice. To exclude these unusable combinations, combinations not satisfying $M^{3/2}\sqrt{\log(M)}/\sqrt{N}\geq 160$ are filtered out. The value of $160$ is empirically determined to filter out the situations with a mode of the estimation error between 0.8 and 1.0. The necessity of this filter is examined via the unfiltered samples. For each combination of $M$ and $N$, kernel density estimates of the model are computed and compared. 

For each resampled model, the model parameters like $\hat{g}$ and $\hat{T}$ are estimated in the simulation process. This allows for calculating the estimation error of these parameters. To find the importance of the errors for different model parameters, we study the relationship between the estimation error of the model and that of the model parameters through scatter plots, Pearson correlation coefficients, and histograms.

\subsubsection{Acceptable estimation error}
With the distribution of the estimation error known, a natural question is: what estimation error is acceptable? A possible purpose of the Expected Threat model is to quantify the ability of a player to create threat within the setting of performance monitoring for a prospective transfer target. We use expert consultation together with simulations to determine when a model is of sufficient quality.

To study the ability of a player to create threat, we consider the values of xT created per 90 minutes player for each player. These performances of players are dependent on the context, like playing position and the league. The xT created by a specific player should thus be interpreted relative to players in a similar situation. Therefore, a natural way of quantifying the quality of players is to consider their quartile within a group of players playing at the same position and in the same league. 

In consultation with practitioners, it is established that an estimation error is acceptable if at most 10\% of the players are assigned an incorrect quartile by the estimated model. Additionally, the quartile assigned by the estimated model may only differ from the ground truth by at most one quartile. To assure this is likely to hold, these situations should be satisfied with at least $0.90$ probability. In short, a model is of acceptable quality if, with probability of $0.90$, less than 10\% of the players are assigned the incorrect quartile with at most one quartile change.

To find to which estimation error this corresponds, we perform simulations. Players playing as centre forwards in the French Ligue 1 in the 2015/2016 season are considered. As these players played in the same competition and position, they roughly faced the same opponents and their values can be compared. Players with less than 300 minutes as a centre forward are filtered out to ensure reliable estimates. This results in a group of 47 players whose ratings can be compared.

A ground truth model with a $16\times12$ ($M=192$) grid is used to simulate estimated models for different values of $N$. Table~\ref{tab: combinations of N max error simulations} describes the different sizes of the training data set $N$, each of which is simulated 10,000 times. For all resampled models, the xT created per 90 minutes is calculated for each player. We then determine the quartile of this player's performance and compare it to the quartile assigned based on the ground truth model. For each estimated model, the estimation error is determined too. This results in a dataset containing 130,000 samples of the estimation error, number of incorrectly assigned player quartiles, and the largest quartile change by a player rating.
\begin{table}[thb]
\centering
\captionsetup{width=0.8\textwidth} 
\caption{Sample sizes \(N\) used in the simulations to find the maximal acceptable estimation error.}
\begin{tabular}{ccc}
\hline
\multicolumn{3}{c}{\textbf{Sample sizes \(N\)}} \\
\hline
100,000   & 630,000   & 3,000,000 \\
130,000   & 1,300,000 & 3,500,000 \\
170,000   & 1,650,000 & 4,000,000 \\
240,000   & 2,000,000 &     	  \\
370,000   & 2,500,000 &           \\
\hline
\end{tabular}
\label{tab: combinations of N max error simulations}
\end{table}
To study the relationship between the estimation error and the number of players with incorrect quartile estimations, the data points are evenly divided into 75 bins with a similar estimation error. For each bin, the median and the 10\%-90\% quantile range are calculated. We use this to find the maximal acceptable size of the estimation error ($EE_{\max}$). This is the highest estimation error that has less than 10\% of the players change at most one quartiles with a 0.90 probability. More specifically, this $EE_{\max}$ is defined as the point where the 90\% quantile of the number of changes is less than 4.7, which represents 10\% of the number of players.

\begin{figure}[hb]
    \centering
    \includegraphics[width=\linewidth]{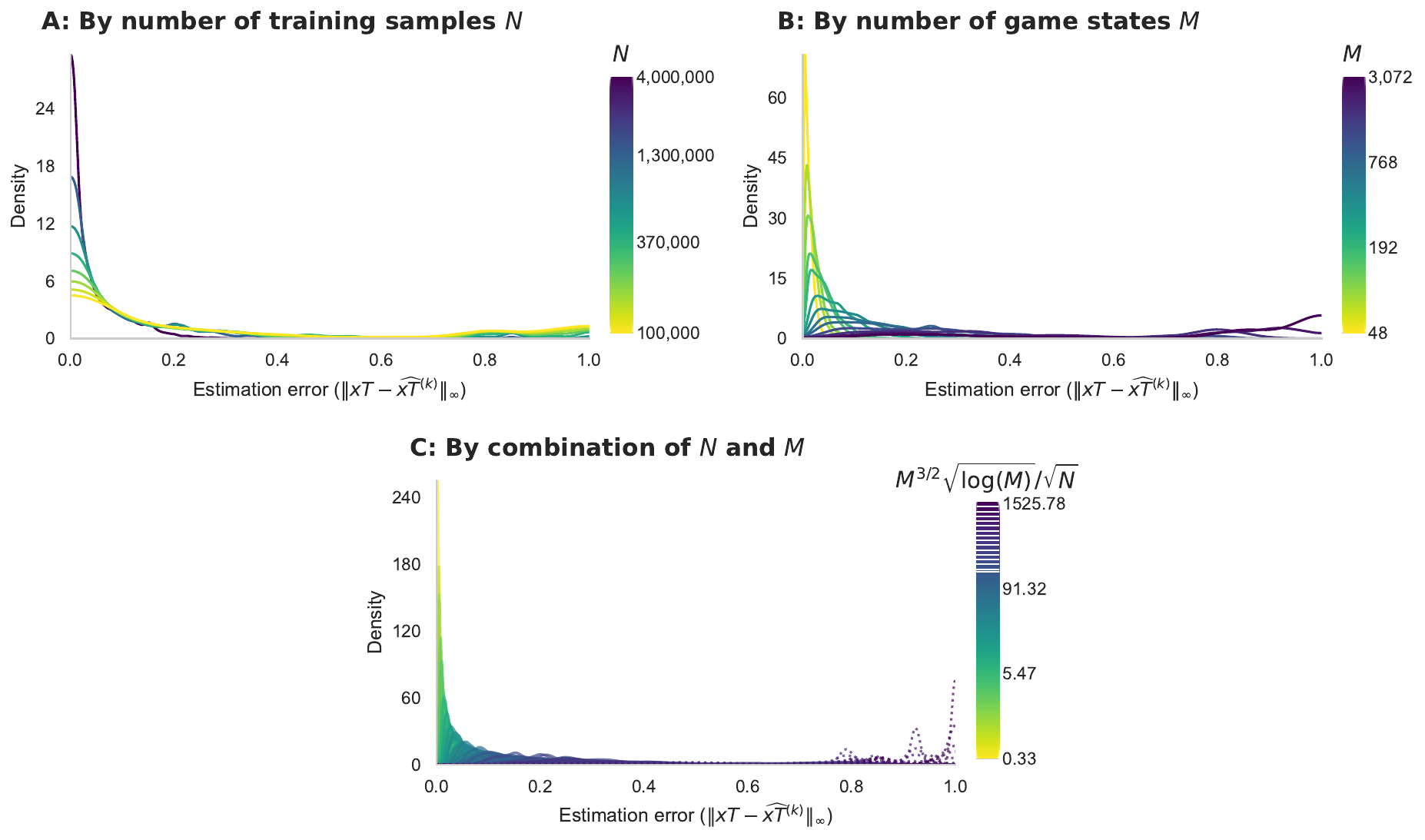}
    \caption{The kernel density estimates of the density of the estimation error for each value of \textbf{A:} $M$, \textbf{B:} $N$, and \textbf{C}: $M^{3/2}\sqrt{\log(M)}/\sqrt{N}$ with filtered out combinations shown as striped lines.}
    \label{fig: KDE plots M and N}
\end{figure}

\subsection{Results}
\subsubsection{Distribution of the estimation error}
Figure~\ref{fig: KDE plots M and N}A and Figure~\ref{fig: KDE plots M and N}B show kernel density estimates for the density of the estimation error based on different values of $N$ and $M$ before any filtering. The results confirm that large error sizes occur more frequently with small values of $N$ and with a large number of game states $M$. This confirms that large estimation errors correspond to situations with a large number of game states $M$ and a small number of training data points $N$ as found in Section~\ref{sec: theoretical results}. 

The density estimates for each combination of $M$ and $N$ is given in Figure~\ref{fig: KDE plots M and N}C. The colours show that indeed a higher value of $M^{3/2}\sqrt{\log(M)}/\sqrt{N}$ generally corresponds to a density with higher estimation errors. The densities show that the cases where $M^{3/2}\sqrt{\log(M)}/\sqrt{N}\geq160$ contain all combinations of $N$ and $M$ where densities have mass above 0.8. This means that the applied filter effectively removes the cases which are infeasible in reality.

The errors between 0.8 and 1.0 correspond to situations where the training data did not contain any shot in a game state with a high probability of scoring. The value of 0 is assigned to the probability of scoring directly from that state $g(s)$. Consequently, the model greatly underestimates the $xT$-values. Because these situations can be observed via the number of shots, it is possible to monitor such large errors for an individual model during model training.

The relationship between the estimation error and error in the estimation of the model parameters $T$ and $g$ is visualised in Figure~\ref{fig: simulations model part relations}, along with histograms of the involved quantities. The histograms indicate that errors in estimating the goal-scoring probabilities, $g$, are generally larger than the errors in estimating the transition matrix weighted by their $xT$-value. Additionally, the error in the goal-scoring probabilities has a high Pearson correlation coefficient ($r=0.9935$) with the estimation error. The correlation coefficient of the weighted error in the transition matrix estimation is lower ($r=0.8069$). The results show that in reality the estimation error is not dominated by the errors in the transition matrix, but by the errors in the goal-scoring probabilities.

\begin{figure}[htb]
    \centering
    \includegraphics[width=0.80\linewidth]{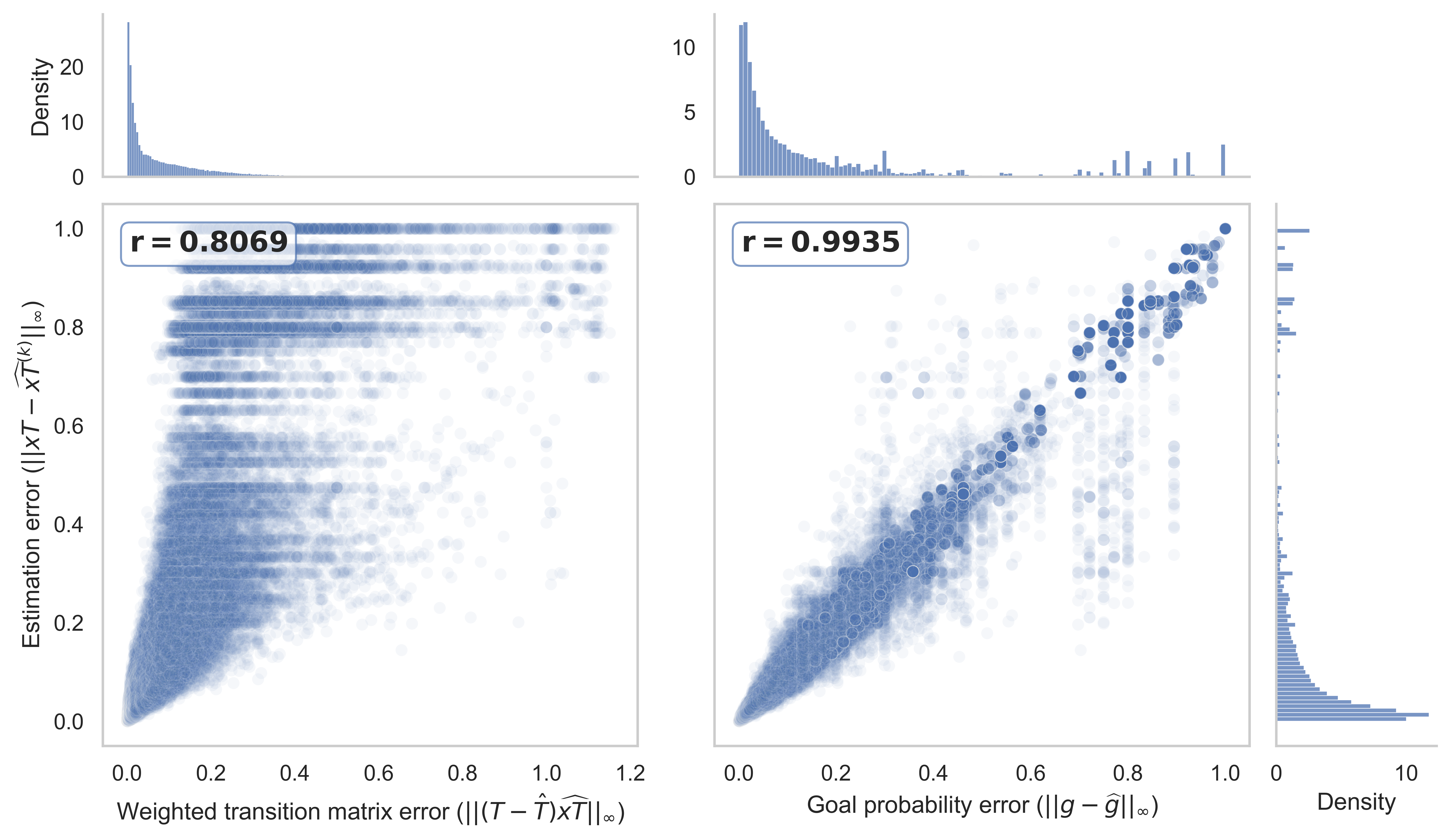}
    \caption{A scatterplot of the error in estimating different parts of the model plotted against the estimation error accompanied by a histogram of their values.}
    \label{fig: simulations model part relations}
\end{figure}

Table~\ref{tab: OLS summary lognormal distribution} shows the results of fitting the parametric model in \eqref{eq simulations lognormal distribution logarithm} using OLS. The $R^2$-value of 0.864 indicates that most of the differences between the estimation errors are explained by the estimated dependency on $M$ and $N$. It suggests that the estimated model captures the relationship well.

\begin{table}[htb]
\caption{Summary of the OLS model fitted on the log estimation error for the data points with $M^{3/2}\sqrt{\log(M)}/\sqrt{N}<160$.
\label{tab: OLS summary lognormal distribution}}
\centering
\resizebox{1.0\columnwidth}{!}{%
\tabcolsep=0pt
\begin{tabular*}{\columnwidth}{@{\extracolsep{\fill}}lccccccc@{}}

\midrule
\textbf{Dep. Variable} & \multicolumn{2}{l}{$\|xT-\widehat{xT}^{(k)}\|_{\infty}$} &
\textbf{R-squared} & 0.864 &
\textbf{No. Observations} & 76,000 & \\

\textbf{Model} & \multicolumn{2}{l}{OLS} &
\textbf{Adj. R-squared} & 0.864 &
\textbf{Log-Likelihood} & $-42,284$ & \\
\toprule
\midrule
\textbf{Parameter} & \textbf{Coef.} & \textbf{Std. Err.} & \textbf{t} & \textbf{P$>|$t$|$} & \textbf{[2.5\%} & \textbf{97.5\%]} & \\

$c$ & $-2.0916$ & 0.017 & $-121.500$ & 0.000 & $-2.125$ & $-2.058$ & \\
$\alpha$ & 1.0100 & 0.002 & 635.048 & 0.000 & 1.007 & 1.013 & \\
$\beta$ & 1.0267 & 0.003 & 402.624 & 0.000 & 1.022 & 1.032 & \\
\midrule
Variance residuals & 0.1782 & & & & & & \\
% \botrule
\end{tabular*}
}
\end{table}

Figure~\ref{fig: residual plots per M and N} shows the residuals plotted against the fitted values for different levels of $M$ and $N$. The results show that small values of $M$ and large values of $N$ have slightly higher variation in the residuals. Still, this effect is relatively small. Thus, the residuals are not too far from homoscedastic.

\begin{figure}[hbt]
    \centering
    \includegraphics[width=\linewidth]{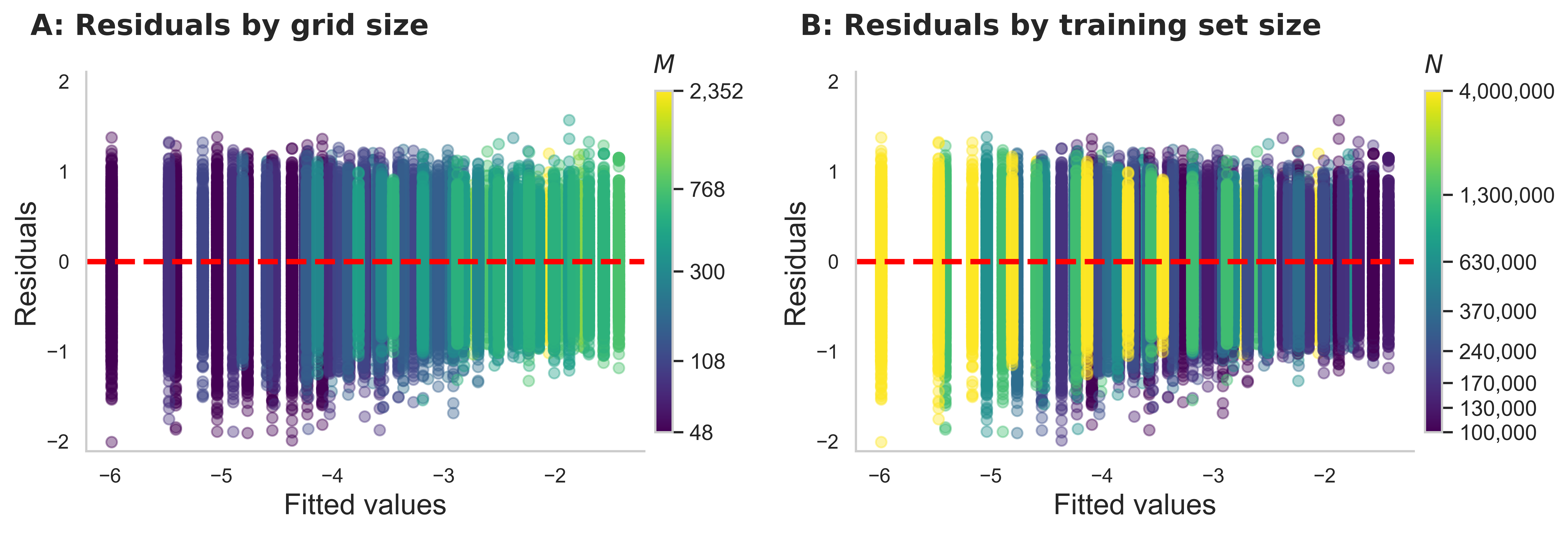}
    \caption{The residuals versus the predicted values in the OLS model coloured by their value of \textbf{A:} $M$ and \textbf{B:} $N$.}
    \label{fig: residual plots per M and N}
\end{figure}

A QQ-plot of the residuals is shown in Figure~\ref{fig: results simulation QQ-plot}. The residuals roughly follow a normal distribution, although the data have a heavy left tail and a light right tail. This means that the fitted model slightly overrepresents large errors compared to reality, which translates to a small overestimation of the quantiles of the estimation errors in practice. As a result, the rules of thumb derived under the normality assumption will be slightly conservative, providing a margin of safety in practice.
\begin{figure}[htb]
    \centering
    \includegraphics[width=0.55\linewidth]{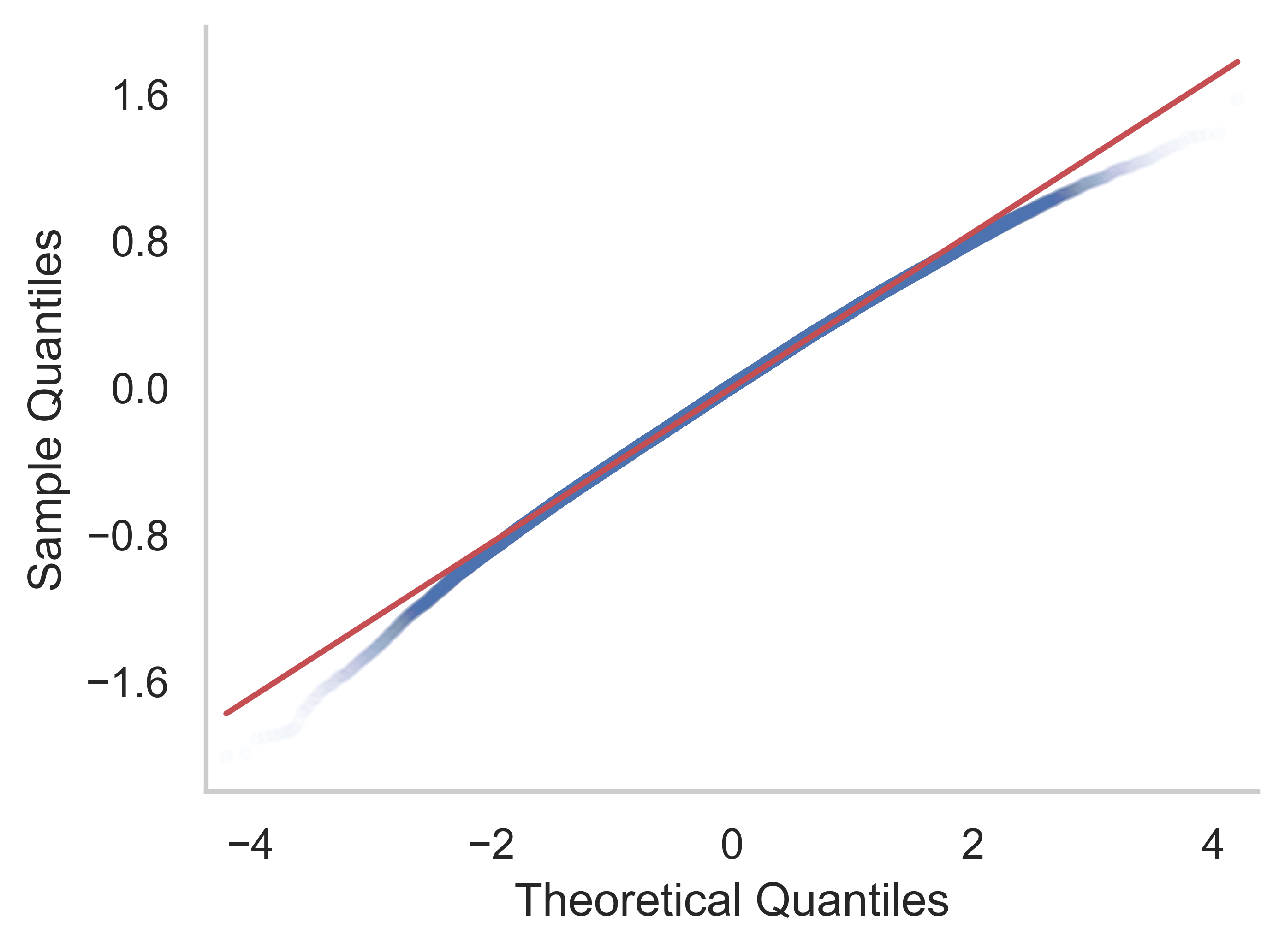}
    \caption{A QQ-plot of the model residuals versus a standard normal distribution.}
    \label{fig: results simulation QQ-plot}
\end{figure}

The results indicate that the model explains most of the variance of the estimation error, that the residuals are roughly homoscedastic and that they are roughly normally distributed. Therefore, we conclude that the estimation error approximately follows a lognormal distribution on the part of the domain of $M$ and $N$ that is used in practice. The values of the coefficients in Table~\ref{tab: OLS summary lognormal distribution} indicate that the size of the estimation error approximately has the following distribution.

\begin{equation}\label{eq: result conclusion approximate distribution estimation error}
    \left\|xT-\widehat{xT}^{(k)}\right\|_\infty = e^{-2.0916}\frac{M^{1.01}}{(\sqrt{N})^{1.0267}}e^{\varepsilon}, \qquad \text{where } \varepsilon\sim N(0, 0.1782).
\end{equation}

\subsubsection{Acceptable estimation error}
Domain expert consultation indicated that a model is of acceptable quality if less than 10\% of the players are assigned the incorrect quartile with at most one quartile change. This should hold with probability 0.9. The simulations indicate that more than one quartile changes rarely occur. The relationship between the size of the estimation error and the number of players assigned the wrong quartile is given in Figure~\ref{fig: quantile changes per estimation error}. It is visible that the number of quartile changes typically increases with the estimation error. 

More specifically, the results show that the highest estimation error having less than 10\% quartile changes with a probability of 0.90 is 0.0192. Therefore, it is found that the maximal acceptable estimation error long term performance monitoring of transfer targets is $EE_{max}=0.0192$.
\begin{figure}[hbt]
    \centering
    \includegraphics[width=0.75\linewidth]{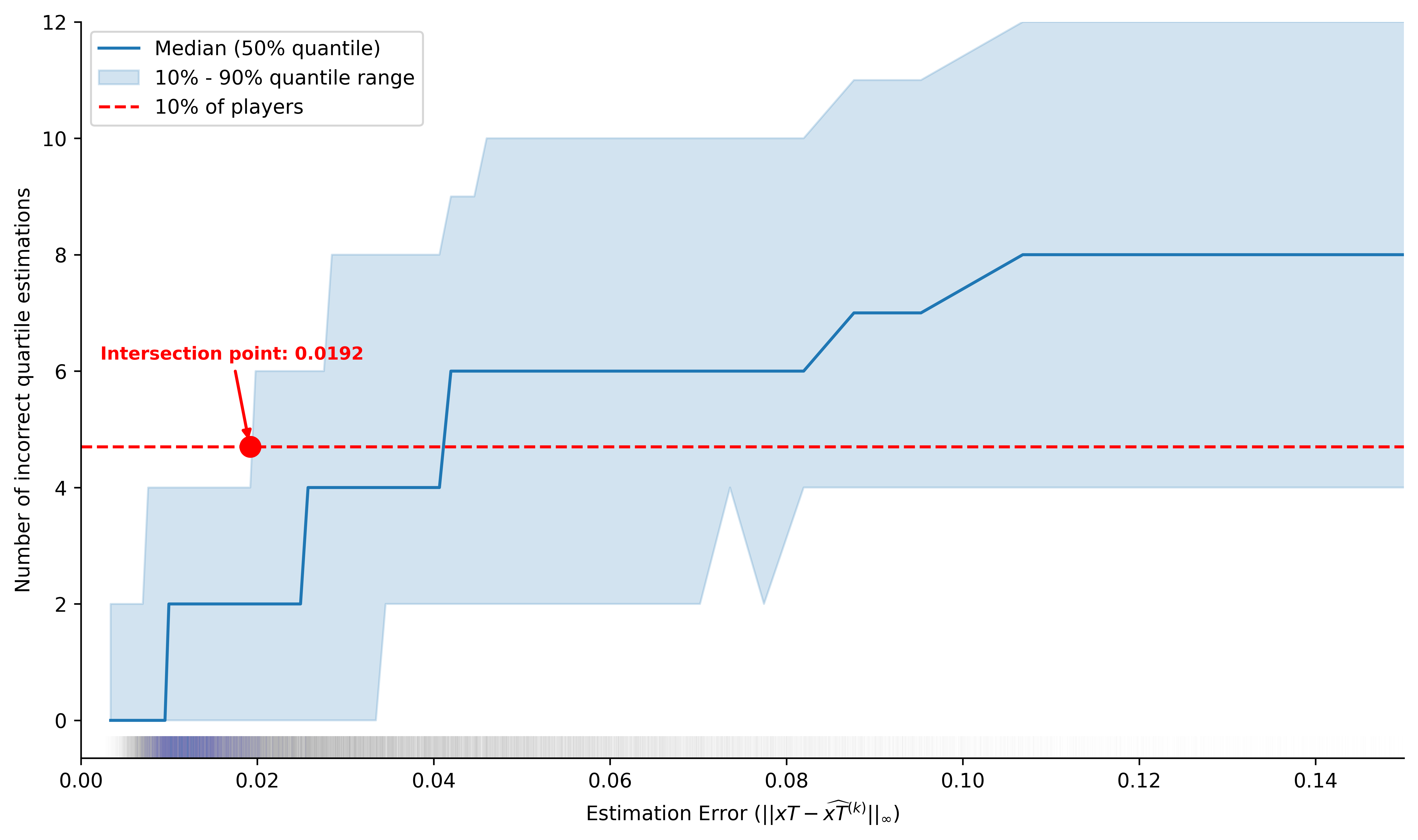}
    \caption{The estimated median and the 10\%-90\% quantile range of the number of incorrect quartile assignments of players by the resampled models for different values of the estimation error.}
    \label{fig: quantile changes per estimation error}
\end{figure}

\section{Application}\label{sec: Application}
The Expected Threat model can be used to assist in decision-making in real-life situations at football clubs. For instance, the model can be used to calculate the xT created per 90 minutes for different players. The xT created per 90 minutes is a performance metric that can be used to objectively compare transfer targets when trying to find a player that is good at moving possession to high-threat situations. 

To responsibly use the Expected Threat model, one first has to ensure that the Expected Threat model is of sufficient quality and has an acceptable estimation error. The findings of this study can be applied either to validate the quality of an existing model or to design a new model of sufficient quality. For this purpose, we formulate rules of thumb based on the results in Section~\ref{sec: Simulations}. Lastly, it is illustrated how an Expected Threat model with an acceptable estimation error can be used for scouting purposes on publicly available StatsBomb data of the 2020 Euros \citep{StatsBombOpenData}.

\subsection{Model validity}

\subsubsection{Quality check of an existing model}
In case football clubs already have access to an existing Expected Threat model, it is beneficial to check whether the existing model is already of sufficient quality. This prevents the unnecessary use of time and resources.

Consider the Expected Threat model by \citet{Singh2018}, visualised in Figure~\ref{fig:singh_xt}B, as an example. This model was trained on a data set of the 2017-2018 Premier League season, which roughly corresponds to $N=620,000$. It has a $16\times12$ grid resulting in $M=192$. These parameters can be plugged into \eqref{eq: result conclusion approximate distribution estimation error} to find that the estimation error behaves like a lognormal variable as visualised in Figure~\ref{fig:singh_xt}A. 

\begin{figure}[htb]
    \centering
    \includegraphics[width=0.9\linewidth]{figures/singh_combined_plot.pdf}
    \caption{An illustration of the quality check of the Expected Threat (xT) model introduced by \citet{Singh2018}. \textbf{A:} The corresponding distribution of the estimation error of such model with with $M=192$ and $N=620{,}000$. \textbf{B:} A visualisation of the $xT$-values of the model.}
    \label{fig:singh_xt}
\end{figure}

The area under the curve with estimation errors smaller than $EE_{max}=0.0192$ represents the probability of the estimation error being of acceptable size. Applying this criterion yields a probability of 0.2209, suggesting that the model may not be of the required quality. Player ratings derived from this model should be interpreted with appropriate caution.

\subsubsection{Selecting grid size for given training set}
If no model of sufficient quality is present, it is necessary to train a new model. The most likely scenario is that there is a specific data set available to train the model on. A natural question that arises is: what is the best number of game states $M$ given a training dataset of size $N$?

When selecting the grid, there is a clear trade-off present. On the one hand, a small number of game states is desirable since it results in a smaller error in estimating the Expected Threat values. On the other hand, a model with more game states is better able to distinguish between different in-game situations, which reduces the error due to model assumptions as discussed in Section~\ref{subsec: quantification of estimation error}. Hence, a finer grid could be favourable. 

Based on the results of the simulations in combination with expert consultation, we formulate the following rule of thumb. \textit{Select the most flexible model (highest $M$) such that the estimation error is smaller than $EE_{max}=0.0192$ with 0.90 probability.} This selection rule ensures that the model is best able to distinguish between different in-game scenarios, while keeping the error within an acceptable range.

As an illustration, we consider a situation where 4 seasons of data are available for a specific league with 20 teams. This corresponds to 2.48 million data points. Using that $N=2,480,000$, Figure~\ref{fig: rule of thumb M}A shows the quantiles of the estimation error for different grid sizes and the maximal acceptable estimation error $EE_{max}$. The rule of thumb indicates to select $M=80$, which corresponds to a $10\times8$ grid. Figure~\ref{fig: rule of thumb M}B shows the model obtained with a training set of the given size and the obtained grid.

\begin{figure}[htb]
    \centering
    \includegraphics[width=0.9\linewidth]{figures/xT_find_grid_combined_plot.pdf}
    \caption{An illustration of the rule of thumb for determining the grid size given $N=2{,}480{,}000$. \textbf{A:} The quantiles of the estimation error of an Expected Threat model for different grid sizes. The point just before the intersection point of the 90\% quantile and $EE_{\max}$ gives the selected grid size. \textbf{B:} A visualisation of the $xT$-values of the resulting model with a $10\times8$ grid trained on $N=2{,}480{,}000$ events.}
    \label{fig: rule of thumb M}
\end{figure}

\subsubsection{Required data set size for a given grid}
When a football club is setting up the infrastructure to use in-house data-driven modelling, it might be the case that a dedicated training data set is not yet fully present. To sufficiently distinguish between in-game situations for a specific application, the staff might desire an Expected Threat model with a specific grid and would like to know what the required size of the data set is.

The findings show that a larger training data set results in a smaller estimation error. However, purchasing access to a larger training data set is more expensive for a football club. Additionally, one might not want to use data from too many different seasons or leagues to ensure alignment between the training data and the application. For these cases, we formulate the following rule of thumb. \textit{Select the smallest value $N$ such that the estimation error is smaller than $EE_{max}=0.0192$ with probability $0.90$.}

Suppose that an Expected Threat model with a grid of $16\times12$ is desired. Using that $M=16\cdot12=192$, the quantiles of the estimation error for different sizes of the training set $N$ are shown in Figure~\ref{fig: rule of thumb N}A. The smallest value of $N$ such that the estimation error is acceptable is 3,348,000. Therefore, a training set of roughly $N=3.35$ million data points is required, corresponding to data of 5.4 seasons. We trained a model with the $16\times12$ grid on a data set of the required size, which is visualised in Figure~\ref{fig: rule of thumb N}B.

\begin{figure}[htb]
    \centering
    \includegraphics[width=0.9\linewidth]{figures/xT_find_sample_size_combined_plot.pdf}
    \caption{An illustration of the rule of thumb for determining the needed data set size given a $16\times12$ grid ($M=192$). \textbf{A:} The quantiles of the estimation error of an Expected Threat model for training set sizes $N$. The intersection point of the 90\% quantile and $EE_{\max}$ is the selected training set size. \textbf{B:} A visualisation of the $xT$-values of the resulting model with a $16\times12$ grid trained on $N=3{,}348{,}000$ events.}
    \label{fig: rule of thumb N}
\end{figure}

\subsection{Talent identification at the 2020 European Championship}
With a validated Expected Threat model, it is possible to evaluate on-the-ball actions of players. The quality of an on-the-ball action can be determined by the increase in threat, given by the estimated probability of scoring a goal, as described in \eqref{eq: diff in xT}. By evaluating the quality of the in-game actions performed by a player, ratings can be obtained describing the on-the-ball quality of a player. 

As an example, we use the Expected Threat model obtained by the last rule of thumb ($M=192$ and $N=3{,}348{,}000$) to find the midfielders who created the most goal-scoring threat at the 2020 European Championship. This can, for instance, be used to aid in a scouting process to identify possible targets when searching for a threat-creating midfielder. 

The data used is part of the StatsBomb Open Data and is publicly available. A filter was applied that selected those who played most games as midfielder and played at least 300 minutes. For each player, ratings were obtained by calculating the total xT difference by his on-the-ball actions normalised per 90 minutes. Within this group of players, the quantiles were calculated for each player.

The results are visible in Figure~\ref{fig: application beeswarm plot midfielders euros}. It shows that the xT created per 90 minutes varies between 0.055 and 0.336. Kai Havertz, Marcel Sabitzer, Mikel Damsgard, and Thomas M\"uller have the largest xT created and are in the fourth quartile (Q4), highlighted with dark green. Because of the relatively small number of games in the tournament, the differences cannot be interpreted as significant, especially not the small difference between Kai Havertz and Marcel Sabitzer. We conclude that these players were among the best at creating goal-scoring threat of all midfielders at the 2020 European Championship.

\begin{figure}[htb]
    \centering
    \includegraphics[width=0.9\linewidth]{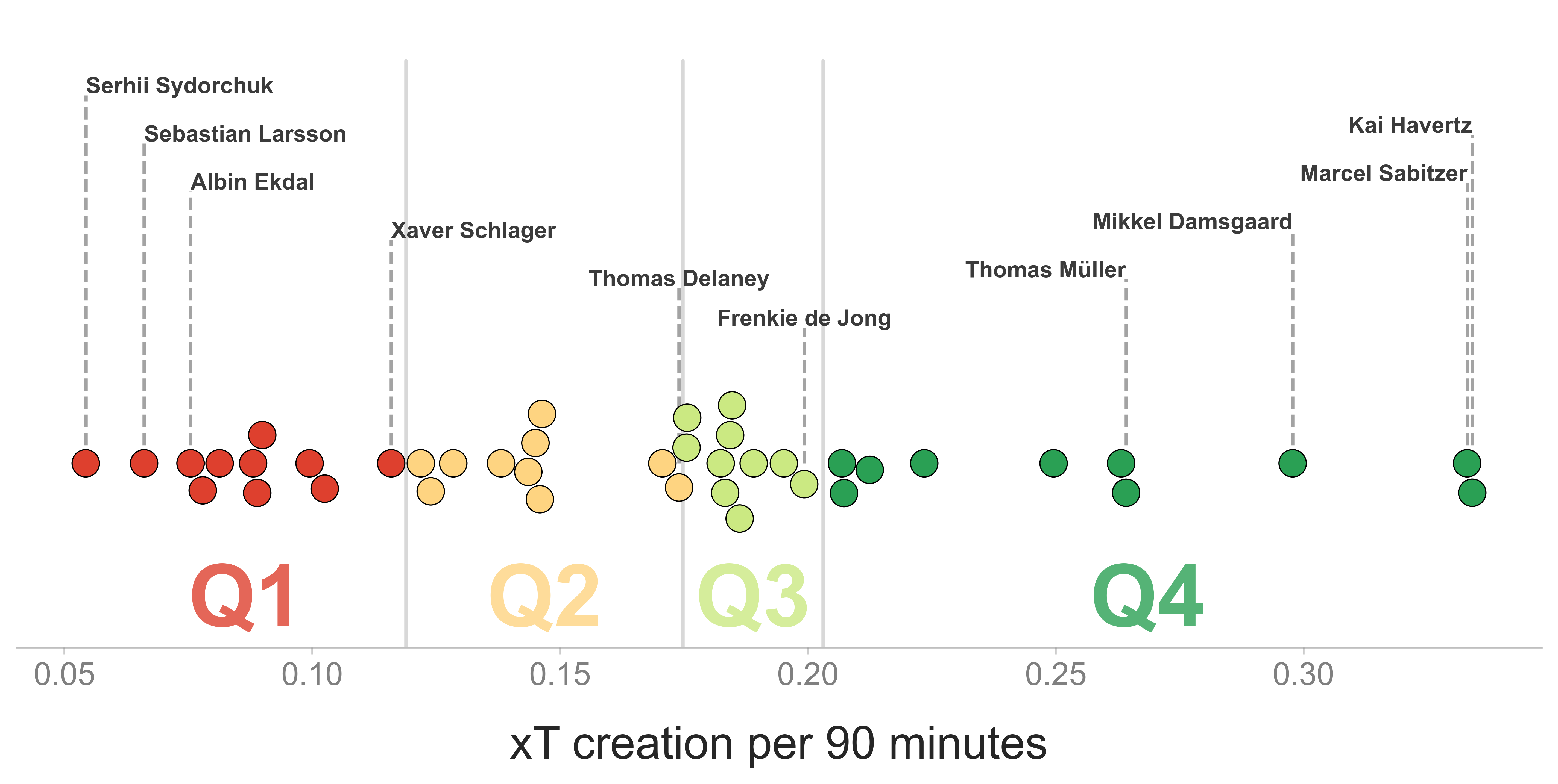}
    \caption{The xT added per 90 minutes by midfielders with at least 300 minutes played at the 2020 European Championship.}
    \label{fig: application beeswarm plot midfielders euros}
\end{figure}

On the other hand, Serhii Sydorchuk, Sebastian Larsson, and Albin Ekdal appear to be the worst-performing players according to this metric. However, possible biases in the data or the model should always be taken into account when evaluating the performances of players. In this specific situation, a selection bias is clearly apparent. Midfielders playing for countries that did not pass the group phase of a tournament most likely underperformed. Still, these were excluded by filtering out those with less than 300 minutes played because no reliable estimate of the xT created per 90 minutes could be made. Thus, the players not being included are likely the worst performing players, and the seemingly low-performing players in Figure~\ref{fig: application beeswarm plot midfielders euros} might have average or even good performance. This highlights the importance of taking into account possible biases in the data.

\section{Discussion}
The goal of this paper was (1) to gain insights into how the error of the Expected Threat model depends on the number of game states $M$ and the number of training points $N$, and (2) to provide guidelines for practitioners to build and apply the Expected Threat model. We find that the estimation error behaves approximately like a lognormal random variable proportional to $M^{1.01}/N^{0.51335}$. In addition, simulations combined with expert consultations identify a maximal acceptable error of $EE_{\max}=0.0192$. Based on this, we established practical rules of thumb to find the appropriate number of game states ($M$) or the desired number of training points ($N$) when constructing a model. An application to the 2020 Euros further illustrates how player performance can be evaluated relative to comparable players and with awareness of potential biases.

The theoretical analysis yields a bound on the estimation error of order $M^{3/2}\sqrt{\log(M)}/\sqrt{N}$, which is dominated by errors in the estimated transition matrix. The simulations show that the estimation error is mostly determined by the errors in the goal-scoring probabilities, which are of a lower order than the errors in the estimated transition matrix. Consequently, the error behaves approximately lognormal like 
\begin{equation*}
    \left\|xT-\widehat{xT}^{(k)}\right\|_\infty = e^{-2.0916}\frac{M^{1.01}}{(\sqrt{N})^{1.0267}}e^\varepsilon, \qquad \text{where } \varepsilon\sim N(0, 0.1782).
\end{equation*}
We conclude that the typical model quality is substantially better than suggested by the theoretical bound.

The estimation error seems to be largely influenced by the estimation of the direct goal-scoring probabilities. High-threat states are rarely observed, because defending teams actively prevent ball possession there. With many game states $M$ and few data points $N$, a lack of observed shots can cause the model to assign zero probability of scoring directly via a shot in these high-value states, causing large errors. Consequently, the most fruitful improvements to the Expected Threat model could be attained by enhancing the estimation of the goal-scoring probabilities. Leveraging existing shooting probability models and xG models to estimate the corresponding quantities will lead to more efficient estimators. Enforcing symmetry, applying smoothing techniques or incorporating prior knowledge could also boost the performance based on domain knowledge as illustrated by \citet{vanArem2024} and \citet{HassaniRamdaniLofti2025}.

The maximal acceptable error ($EE_{\max}$) of 0.0192 means that the absolute error of the $xT$ for each game state must be smaller than 1.92\% for applications in player scouting. Existing models, such as that of \citet{Singh2018}, might fall short of this threshold. 
Notably, the rules of thumb suggests that approximately 5.4 seasons of data are required for a $16\times12$ grid, far more than typically used in practice. Additional model validation methods like visual inspection by domain experts and verification methods such as described by \citet{DavisBransen2024} can be performed to complement the rules of thumb described. Sufficient technical expertise is needed to perform this. Under these circumstances, less conservative values might be considered instead of $EE_{\max}=0.0192$. This highlights the need for caution when applying the Expected Threat or related possession value models.

In the simulations, models trained on an open-access data set of roughly 4 million data points were assumed as the ground truth, which constitutes a limitation. These models can contain estimation errors and may be less smooth than the true underlying values. This might make the estimation task more difficult, and consequently, the estimation error may be slightly overestimated. Future research with larger, probably proprietary, data sets could evaluate the exact impact of this effect.

Within the broader set of Expected Possession Value models \citep{Fernandez2019, Fernandez2021, Stockl2021, Overmeer2025}, there exist other models that incorporate additional in-game context. Like the Expected Threat model, these EPV models implicitly rely on the Markov assumption. The simulations and rules of thumb in this paper can be extended in a similar manner to such models as well, enabling a systematic assessment of model quality. In this way, this paper provides a framework to quantify model quality, despite the ground truth being an intangible object in football analytics.

\section{Conflicts of interest}
The authors declare that they have no competing interests.

% \section{Funding}
% This work is supported in part by funds from the National Funding Body (NFB: \# XXXXXXX and \# YYYYYYY).

\section{Data availability}
The data underlying this article are available in the the Statsbomb Open Data set \citep{StatsBombOpenData} at \url{https://github.com/statsbomb/open-data}. The code for Expected Threat model with methodology to perform the simulations is provided by the authors at \url{https://github.com/kvanarem42/xthreat-error-quantification} along with the code to reproduce the study.

\section{Author contributions statement}
K.v.A., J.S., M.B. and G.J. conceptualised the problem statement. K.v.A. carried out the theoretical analysis, simulations and visualisations. K.v.A., with input of M.B., translated the results to the application. J.S. and G.J. provided supervision. K.v.A. wrote the draft. K.v.A., J.S., M.B. and G.J. reviewed and edited the manuscript.
% Must include all authors, identified by initials, for example:
% S.R. and D.A. conceived the experiment(s),  S.R. conducted the experiment(s), S.R. and D.A. analysed the results.  S.R. and D.A. wrote and reviewed the manuscript.

\section{Acknowledgments}
The authors thank StatsBomb for making the StatsBomb Open Data set available.
\bibliographystyle{oup-abbrvnat}
\bibliography{reference}

% %% sample for biography with author's image
% \begin{biography}{{\color{black!20}\rule{77pt}{57pt}}}{\author{Author Name.} This is sample author biography text. The values provided in the optional argument are meant for sample purposes. There is no need to include the width and height of an image in the optional argument for live articles. This is sample author biography text this is sample author biography text this is sample author biography text this is sample author biography text this is sample author biography text this is sample author biography text this is sample author biography text this is sample author biography text.}
% \end{biography}

% %% sample for biography without author's image
% \begin{biography}{}{\author{Author Name.} This is sample author biography text this is sample author biography text this is sample author biography text this is sample author biography text this is sample author biography text this is sample author biography text this is sample author biography text this is sample author biography text.}
% \end{biography}

%%%%%%%%%%%%%%

\begin{appendices}

\section{Proofs of propositions and theorems}
The propositions and theorems in this paper were formulated with the $\ell_\infty$-norm. Most of them also hold for general vector norms with corresponding vector-induced matrix norms. The proofs below do not specify the norm if they hold in this general setting.
\begin{proof}[Proof of Proposition~\ref{prop: truncation error vectorized}]
    The $k$'th iteration of Algorithm~\ref{alg: iterative algorithm vectorized} can alternatively be formulated as
    \begin{equation}
        \widehat{xT}^{(k)} = \sum_{\ell=0}^{k-1}T^\ell g.
    \end{equation}
    Because $||\hat{T}||<1$, a $\widehat{xT}$ can similarly be formulated as
    \begin{equation}\label{eq: geometric formulation xT}
        \widehat{xT} = \sum_{\ell=0}^{\infty}T^\ell g.
    \end{equation}
    With these expressions, it follows that
    \begin{equation*}
        \widehat{xT} - \widehat{xT}^{(k)} = \sum_{\ell=k}^\infty \hat{T}^\ell \hat{g} = \hat{T}^k\sum_{\ell=0}^\infty \hat{T}^\ell\hat{g}.
    \end{equation*}
    Since $||\hat{T}||<1$, the geometric series give that
    \begin{equation*}
        ||\widehat{xT} - \widehat{xT}^{(k)}|| \leq ||\hat{T}||^k\sum_{\ell=0}^\infty||\hat{T}||^\ell||\hat{g}|| = \frac{||\hat{g}||\ ||\hat{T}||^k}{1-||\hat{T}||}.
    \end{equation*}
\end{proof}

\begin{proof}[Proof of Proposition~\ref{prop: split of statistical error}]
    By definition, \eqref{eq: characteristic equation vectorized} holds for both $xT$ and $\widehat{xT}$. Consequently, 
    \begin{align*}
        ||xT-\widehat{xT}|| &\leq ||g - \hat{g}|| + ||T\cdot xT- \hat{T}\cdot \widehat{xT}||\\
        &\leq ||g - \hat{g}|| + ||T||\ ||xT - \widehat{xT}|| + ||(T-\hat{T})\widehat{xT}||.
    \end{align*}
    Because $||T||<1$, this is equivalent to \eqref{eq: split of statistical error}.
\end{proof}

\begin{proof}[Proof of Proposition~\ref{prop: statistical error goal vector}]
    Denote $\varepsilon=\sqrt{\frac{\log(2M/\alpha)}{2\bar{N}_s}}$. The Hoeffding's inequality states that 
    \begin{equation*}
        P(|g-\hat{g}(i)|\geq \varepsilon) \leq 2\exp(-2\bar{N}_s\varepsilon^2).
    \end{equation*}
    Consequently,
    \begin{align*}
        P(||g-\hat{g}||_\infty\geq \varepsilon) & =P\left(\max_{i=1,\ldots,M} |g(i)-\hat{g}(i)| \geq \varepsilon\right)\\
        &= P\left(\bigcup_{i=1,\ldots,M} \{|g(i)-\hat{g}(i)|\geq \varepsilon\}\right)\\
        &\leq \sum_{i=1}^M P(|g(i)-\hat{g}(i)|\geq \varepsilon)\\
        &\leq 2M \exp(-2\bar{N}_g\varepsilon^2)\\
        &=\alpha.
    \end{align*}
    Consequently, the inequality holds with probability $1-\alpha$. The equality of the bounds follows from the definition of $\bar{N}_g$.
\end{proof}

\begin{proof}[Proof of Proposition~\ref{prop: statistical error transition matrix}]
    Denote $\varepsilon=M\sqrt{\frac{\log(2M^2/\alpha)}{2\bar{N}_m}}$. Note that Hoeffding's inequality gives  
    \begin{equation*}
        P(|T_{i,j} - \hat{T}_{i,j}| \geq \varepsilon/M) \leq 2\exp(-2\bar{N}_m\varepsilon^2/M^2)
    \end{equation*}
    Then,
    \begin{align*}
        P(||T-\hat{T}||_\infty\geq \varepsilon) &\leq P\left(M \max_{i,j=1\ldots,M} |T_{i,j}-\hat{T}_{i,j}|\geq \varepsilon\right)\\
        &=P\left(\bigcup_{i,j=1,\ldots,M} M|T_{i,j}-\hat{T}_{i,j}|\geq \varepsilon\right)\\
        &\leq \sum_{i,j=1,\ldots,M} P(|T_{i,j}-\hat{T}_{i,j}|\geq \varepsilon/M)\\
        &\leq \sum_{i,j=1,\ldots,M} 2\exp(-2\bar{N}_m\varepsilon^2/M^2)\\
        &= 2M^2\exp(-2\bar{N}_m\varepsilon^2/M^2)\\
        &=\alpha.
    \end{align*}
    Consequently, the inequality holds with probability $1-\alpha$. The equality of the bounds follows from the definition of $\bar{N}_m$.
\end{proof}

\begin{proof}[Proof of Theorem~\ref{thm: estimation error}]
    The inequality follows directly from combining Proposition~\ref{prop: truncation error vectorized}, Proposition~\ref{prop: split of statistical error}, Proposition~\ref{prop: statistical error goal vector}, and Proposition~\ref{prop: statistical error transition matrix}. Hereby, it is used that $||(T-\hat{T})\widehat{xT}||_\infty\leq ||T-\hat{T}||_\infty$.
\end{proof}

\end{appendices}

\end{document}